\documentclass{article}
\usepackage{graphicx}
\usepackage{amsmath,amssymb,amsthm}
\usepackage{amsfonts}
\usepackage{currfile}
\usepackage{floatrow}
\usepackage{authblk}

\newtheorem{theorem}{Theorem}[section]

\newtheorem{conjecture}[theorem]{Conjecture}
\newtheorem{problem}[theorem]{Problem}
\newtheorem{lemma}[theorem]{Lemma}

\theoremstyle{definition}
\newtheorem{defn}{Definition}[section]

\theoremstyle{remark}
\newtheorem*{rem}{Remark}

\makeatletter
\renewcommand\@biblabel[1]{}
\makeatother

\begin{document}
\title{Edge-cuts Optimized for Average Weight: a new alternative to Ford and Fulkerson \footnote{An alternate version of this material was published in
Asia-Pacific Journal of Operational Research, Vol. 36, No. 2 (2019) 1940006 under the title {\em ``Edge-cuts of Optimal Average Weights''}}\footnote{This work is partially supported by an  
NIH Grant 1-R01-DC015901 and an NSF grant DMS-1700218}
}
%\author{Scott Payne, Edgar Fuller, Cun-Quan Zhang}
\date{}
\author[1]{Scott Payne\thanks{spayne7@mix.wvu.edu}}
\author[2]{Edgar Fuller\thanks{efuller@fiu.edu}}
\author[1]{Cun-Quan Zhang\thanks{cqzhang@mail.wvu.edu}}
\affil[1]{Department of Mathematics, West Virginia University \newline Morgantown, WV 28506, USA}
\affil[2]{Department of Mathematics, Florida International University \newline Miami, FL 33199, USA}

\renewcommand\Authands{ and }

\maketitle

% \footnote{file name: \currfilename}

\begin{abstract}
Let $G$ be a directed graph associated with a weight $w: E(G) \rightarrow R^+$. For an edge-cut $Q$ of $G$, the average weight of $Q$ is denoted and defined as $w_{ave}(Q)=\frac{\sum_{e\in Q}w(e)}{|Q|}$. An edge-cut of optimal average weight is an edge-cut $Q$ such that $w_{ave}(Q)$ is maximum among all edge-cuts (or minimum, symmetrically).
 In this paper, a polynomial algorithm for this problem is proved for finding such an optimal edge-cut in a rooted tree, separating the root and the set of all leafs.
This algorithm
  enables us to develop an automatic clustering method with more accurate detection of communities embedded in a hierarchy tree structure.
\end{abstract}

\section{Introduction}
Max-Flow-Min-Cut is one of the oldest optimization problems in network theory~\cite{Ford1956}. It is also known that solving the max cut problem is NP complete~\cite{Garey1979,Karp1972}. In this paper, we define a new optimization problem for finding edge cuts in a special class of weighted directed graphs, specifically rooted weighted trees. These types of trees are typically used to represent relationships in hierarchical data structures. The minimum cut problem optimizes the sum of weights of edges in a given cut, over all cuts separating source and sink. Our new problem instead optimizes the average of the weights of edges in a given cut, over all root-separating cuts. Importantly, we propose a polynomial time algorithm for solving this optimization problem. Furthermore, the inequalities in the algorithm and proof may be adapted to solve either the minimization problem  {\it or} the maximization problem. Here we present the version for maximization, the version for minimization is similar and left to the reader.

\subsection{Applications for data mining}
In the development of clustering methods, one of the most challenging problems is identifying the level of association among data points, such as the vertices of a graph, that provides "a best" output of communities. It is known that for complex data there is no one optimization problem that can characterize a true optimal for all data sets \cite{Peele2017}, however in practice it is often necessary to find an output that is as good as possible according to some established criteria. In the case of hierarchical clustering, this question becomes the determination of which collection of cuts along the tree in the hierarchical dendrogram will be selected to form the set of communities as the final output. It has been observed that {\it ``There are no completely satisfactory algorithms that can be used for determining the number of population clusters for many types of cluster analysis"}~\cite{SAS08}.

In \cite{Qi2014}, the minimum-cut approach was introduced for the automatic selection of the final output of communities. Here we provide an alternate type of edge cut optimization problem for use in methods such as that of \cite{Qi2014}. We have observed that in many data sets with hierarchical structure, community structures at lower levels (i.e. larger number of communities) might not be discovered when the edge cut optimization is based on minimum-cut (which minimizes a sum of edge weights) as less elements of the sum might return a smaller value and hence less edges in the cut might return a more minimal sum. By optimizing an average of weights in an edge cut as opposed to the sum, edge-cut community detection may more freely return optimal solutions at lower levels, often referred to as ``high resolution'' clusters. This type of high resolution clustering has become an increasing focus in the neuroscientific field of connectomics where neuron to neuron synaptic connectivity data may naturally have a very large number of local clusters.

\subsection{Notation and Definitions}

A rooted tree  $T$ is  a directed graph whose underlying graph is a tree and, there is a given vertex $v_0$, called the root, such that, for every vertex $x \in V(T)$, the unique path of the tree
 from $v_0$ to $x$ is a directed path.

Let $T$ be a rooted, weighted tree with edge weight function $w: E(T) \rightarrow \mathbb{R}^+$ assigning a weight $w(e)$ to each edge $e\in T$. Let $v_0$ be the root of $T$. For any edge set $X \subseteq E(T)$ we set $w (X) := \sum_{e \in X} w(e)$ to be the sum of the weights of the edges in this subset. We denote by $E^+(v)$ the out edges of vertex $v$ and $E^+(e)$ to refer to the out edges of the vertex that is the head of edge $e$ when this notation is convenient. Set $\alpha_0 := \frac{w (E^+(v_0))}{|E^+(v_0)|}$. When we refer to edge cuts we will usually use the symbol $Q$. {\it Furthermore}, all the edge cuts discussed are cuts separating the root $v_0$ from the set $L$ of {\it leaf vertices} of $T$, we may refer to such cuts as {\it root-separating}. The algorithm presented here relies on the graph operation {\it edge contraction} which we denote using the standard notation $T/e$ when we contract the edge $e$ in the digraph $T$.

\begin{defn}
\label{contractability}
For $e \in E(T - L)$ we define the \textbf{contractibility of $e$}, denoted by $\lambda(e)$, with the following formula:
\[  \lambda(e) = \frac{w(E^+(e)) - w(e)}{|E^+(e)| - 1}
\]
\end{defn}
\begin{defn} We say an edge $e$ is {\it contractible} if $\lambda(e) > \alpha_0$.
\end{defn}
\begin{defn} The edge set $E(T - L)$ may be ordered $\lambda(e_1) \geqslant \lambda(e_2) \geqslant ... \geqslant \lambda(e_m)$. We may refer to this ordering as the {\it contractibility ordering}.
\end{defn}

\section{Optimization Problem and Algorithm}

\medskip \noindent{\bf Input.}
A rooted weighted tree $T$ with edge weight $w: E(T) \rightarrow \mathbb{R}^+$ and the root $v_0$ and set of leaves $L$.

\medskip \noindent{\bf Output.}
An edge cut $Q$ of $T$ separating the root $v_0$ and the set $L$ of
 leaves such that
 $\frac{\sum_{e \in Q}w(e)}{|Q|}$ is maximum among all such edge cuts.
\vspace{0.5 in}\\
\medskip \noindent The algorithm is then as follows.
\vskip .1in \noindent
{\bf Step 1.}
Determine
$$\alpha_0 = \frac{w(E^+(v_0))}{|E^+(v_0)|}.$$

\medskip \noindent
{\bf Step 2.}
Sort the edges $e_i$ of the $E(T-L)$ so that
$$\lambda(e_1) \geq \lambda(e_2) \geq ... \geq \lambda(e_m)$$
where $\lambda(e_i)$ is the {\bf contractibility of the edge $e_i$} as in definition \ref{contractability}.

\medskip \noindent
{\bf Step 3.}
If $\lambda(e_1) > \alpha_0$ then
\begin{quote}
\begin{enumerate}
\item[(i)]Denote the {\it in edge} to $e_1$ by $e^*$. Contract $T \leftarrow T/e_1$, and
\item[(ii)] update $\lambda$ value for $e^*$, or update $\alpha_0$ if $e_1$ had no {\it in edge} (it was in $E^+(v_0)$), and
\item[(iii)] repeat Step 2.
\end{enumerate}
\end{quote}
\medskip \noindent
If $\lambda(e_1) \leq \alpha_0$ then go to the END STEP.

\medskip \noindent
{\bf END STEP.} Output: $Q=E^+(v_0)$.

\begin{rem} The output $Q$ above is an edge set of the contracted graph resulting from the running of the algorithm, however $Q$ is also a subset of the original set of edges input to the algorithm. It is in the context of the input graph $T$ that the set $Q$ is the solution to the optimization problem presented here.
\end{rem}

Figures~\ref{FIG: input} and~\ref{FIG: output} illustrate an example of the output of this algorithm.

\setlength{\unitlength}{0.06cm}
\begin{figure}

%%%%%%%%%%%%%%%%%%%%%%%%INOUT%%%%%%%%%%%%%%%%%%%
%%%%%%%%%%%%%%%%%%%%%%%%%%%%%%%%%%%%%%%%%%%%%%%%%%
% \setlength{\unitlength}{0.08cm}

\begin{center}

\begin{picture}(140,120)

%%% lowest left %%
\put(0,0){\circle*{2.6}}
\put(5,0){\circle*{2.6}}
\put(10,0){\circle*{2.6}}
\put(15,0){\circle*{2.6}}
\put(20,0){\circle*{2.6}}
\put(25,0){\circle*{2.6}}
\put(30,0){\circle*{2.6}}
\put(35,0){\circle*{2.6}}
\put(40,0){\circle*{2.6}}

\put(-2,10){$3$}
\put(2,10){$3$}
\put(6,10){$3$}

\put(13,10){$3$}
\put(17,10){$3$}
\put(22,10){$3$}

\put(28,10){$3$}
\put(32,10){$3$}
\put(36,10){$3$}

%%% lowest center %%
\put(50,0){\circle*{2.6}}
\put(55,0){\circle*{2.6}}
\put(60,0){\circle*{2.6}}
\put(65,0){\circle*{2.6}}
\put(70,0){\circle*{2.6}}
\put(75,0){\circle*{2.6}}
\put(80,0){\circle*{2.6}}
\put(85,0){\circle*{2.6}}
\put(90,0){\circle*{2.6}}

\put(48,10){$1$}
\put(52,10){$1$}
\put(56,10){$1$}

\put(63,10){$1$}
\put(67,10){$1$}
\put(72,10){$1$}

\put(78,10){$1$}
\put(82,10){$1$}
\put(86,10){$1$}

%%% lowest right %%
\put(100,0){\circle*{2.6}}
\put(105,0){\circle*{2.6}}
\put(110,0){\circle*{2.6}}
\put(115,0){\circle*{2.6}}
\put(120,0){\circle*{2.6}}
\put(125,0){\circle*{2.6}}
\put(130,0){\circle*{2.6}}
\put(135,0){\circle*{2.6}}
\put(140,0){\circle*{2.6}}

\put(98,10){$2$}
\put(102,10){$2$}
\put(106,10){$2$}

\put(113,10){$2$}
\put(117,10){$2$}
\put(122,10){$2$}

\put(128,10){$2$}
\put(132,10){$2$}
\put(136,10){$2$}

%%% middle level
\put(5,40){\circle*{2.6}}
\put(20,40){\circle*{2.6}}
\put(35,40){\circle*{2.6}}
\put(55,40){\circle*{2.6}}
\put(70,40){\circle*{2.6}}
\put(85,40){\circle*{2.6}}
\put(105,40){\circle*{2.6}}
\put(120,40){\circle*{2.6}}
\put(135,40){\circle*{2.6}}

\put(4,50){$2$}
\put(16,50){$2$}
\put(32,50){$2$}

\put(54,50){$3$}
\put(66,50){$3$}
\put(82,50){$3$}

\put(104,50){$1$}
\put(116,50){$1$}
\put(132,50){$1$}

%%% 2nd highest level
\put(20,80){\circle*{2.6}}
\put(70,80){\circle*{2.6}}
\put(120,80){\circle*{2.6}}

\put(24,90){$1$}
\put(67,90){$2$}
\put(115,90){$3$}

%%% the root
\put(70,120){\circle*{2.6}}
\put(64,122){$v_0$}

%%%% LINES %%%%%
%%% line: lowest, left %%%
\qbezier(5,40)(5,40)(0,0)
\qbezier(5,40)(5,40)(5,0)
\qbezier(5,40)(5,40)(10,0)

\qbezier(20,40)(20,40)(15,0)
\qbezier(20,40)(20,40)(20,0)
\qbezier(20,40)(20,40)(25,0)

\qbezier(35,40)(35,40)(30,0)
\qbezier(35,40)(35,40)(35,0)
\qbezier(35,40)(35,40)(40,0)

%%% line: lowest, center %%%
\qbezier(55,40)(55,40)(50,0)
\qbezier(55,40)(55,40)(55,0)
\qbezier(55,40)(55,40)(60,0)

\qbezier(70,40)(70,40)(65,0)
\qbezier(70,40)(70,40)(70,0)
\qbezier(70,40)(70,40)(75,0)

\qbezier(85,40)(85,40)(80,0)
\qbezier(85,40)(85,40)(85,0)
\qbezier(85,40)(85,40)(90,0)

%%% line: lowest, right %%%
\qbezier(105,40)(105,40)(100,0)
\qbezier(105,40)(105,40)(105,0)
\qbezier(105,40)(105,40)(110,0)

\qbezier(120,40)(120,40)(115,0)
\qbezier(120,40)(120,40)(120,0)
\qbezier(120,40)(120,40)(125,0)

\qbezier(135,40)(135,40)(130,0)
\qbezier(135,40)(135,40)(135,0)
\qbezier(135,40)(135,40)(140,0)

%%% line: middle %%%
\qbezier(20,80)(20,80)(5,40)
\qbezier(20,80)(20,80)(20,40)
\qbezier(20,80)(20,80)(35,40)

\qbezier(70,80)(70,80)(55,40)
\qbezier(70,80)(70,80)(70,40)
\qbezier(70,80)(70,80)(85,40)

\qbezier(120,80)(120,80)(105,40)
\qbezier(120,80)(120,80)(120,40)
\qbezier(120,80)(120,80)(135,40)

%%% line: highest %%%
\qbezier(70,120)(70,120)(20,80)
\qbezier(70,120)(70,120)(70,80)
\qbezier(70,120)(70,120)(120,80)

\end{picture}
\end{center}
\caption{\small\it The input: a weighted tree with the root $v_0$ }
\label{FIG: input}
\end{figure}
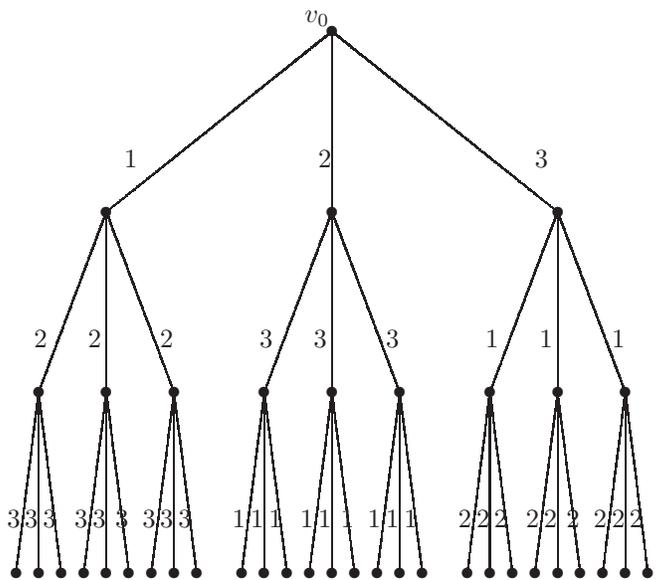

%%%%%%%%%%%%%%%%%%%%%%%%%OUTPUT%%%%%%%%%%%%%%%%%%%%
%%%%%%%%%%%%%%%%%%%%%%%%%%%%%%%%%%%%%%%%%%%%%%%%%%
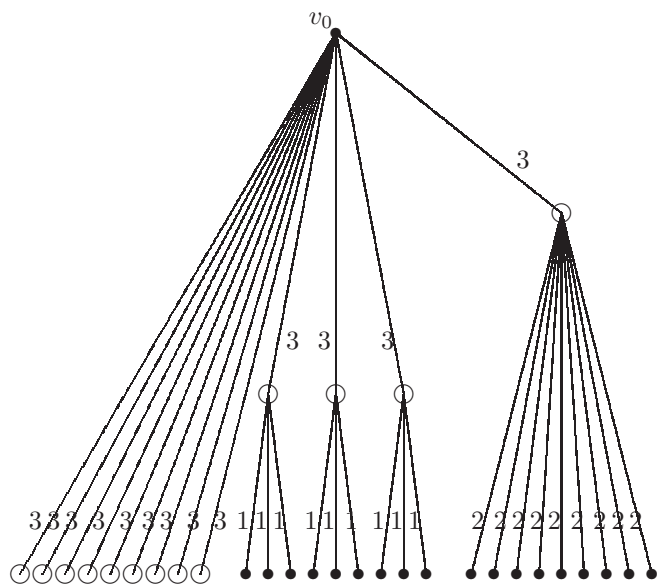
\begin{figure}
\begin{center}

\begin{picture}(140,120)

%%% lowest left %%
\put(0,0){\circle{4}}
\put(5,0){\circle{4}}
\put(10,0){\circle{4}}
\put(15,0){\circle{4}}
\put(20,0){\circle{4}}
\put(25,0){\circle{4}}
\put(30,0){\circle{4}}
\put(35,0){\circle{4}}
\put(40,0){\circle{4}}

\put(2,10){$3$}
\put(6,10){$3$}
\put(10,10){$3$}

\put(16,10){$3$}
\put(22,10){$3$}
\put(27,10){$3$}

\put(31,10){$3$}
\put(37,10){$3$}
\put(43,10){$3$}

%%% lowest center %%
\put(50,0){\circle*{2.6}}
\put(55,0){\circle*{2.6}}
\put(60,0){\circle*{2.6}}
\put(65,0){\circle*{2.6}}
\put(70,0){\circle*{2.6}}
\put(75,0){\circle*{2.6}}
\put(80,0){\circle*{2.6}}
\put(85,0){\circle*{2.6}}
\put(90,0){\circle*{2.6}}

\put(48,10){$1$}
\put(52,10){$1$}
\put(56,10){$1$}

\put(63,10){$1$}
\put(67,10){$1$}
\put(72,10){$1$}

\put(78,10){$1$}
\put(82,10){$1$}
\put(86,10){$1$}

%%% lowest right %%
\put(100,0){\circle*{2.6}}
\put(105,0){\circle*{2.6}}
\put(110,0){\circle*{2.6}}
\put(115,0){\circle*{2.6}}
\put(120,0){\circle*{2.6}}
\put(125,0){\circle*{2.6}}
\put(130,0){\circle*{2.6}}
\put(135,0){\circle*{2.6}}
\put(140,0){\circle*{2.6}}

\put(100,10){$2$}
\put(105,10){$2$}
\put(109,10){$2$}

\put(113,10){$2$}
\put(117,10){$2$}
\put(122,10){$2$}

\put(127,10){$2$}
\put(131,10){$2$}
\put(135,10){$2$}

%%% middle level

\put(55,40){\circle{4}}
\put(70,40){\circle{4}}
\put(85,40){\circle{4}}

\put(59,50){$3$}
\put(66,50){$3$}
\put(80,50){$3$}

%%% 2nd highest level
\put(120,80){\circle{4}}

\put(110,90){$3$}

%%% the root
\put(70,120){\circle*{2.6}}
\put(64,122){$v_0$}

%%%% LINES %%%%%
%%% line: lowest, left %%%
\qbezier(70,120)(70,120)(0,0)
\qbezier(70,120)(70,120)(5,0)
\qbezier(70,120)(70,120)(10,0)

\qbezier(70,120)(70,120)(15,0)
\qbezier(70,120)(70,120)(20,0)
\qbezier(70,120)(70,120)(25,0)

\qbezier(70,120)(70,120)(30,0)
\qbezier(70,120)(70,120)(35,0)
\qbezier(70,120)(70,120)(40,0)

%%% line: lowest, center %%%
\qbezier(55,40)(55,40)(50,0)
\qbezier(55,40)(55,40)(55,0)
\qbezier(55,40)(55,40)(60,0)

\qbezier(70,40)(70,40)(65,0)
\qbezier(70,40)(70,40)(70,0)
\qbezier(70,40)(70,40)(75,0)

\qbezier(85,40)(85,40)(80,0)
\qbezier(85,40)(85,40)(85,0)
\qbezier(85,40)(85,40)(90,0)

%%% line: lowest, right %%%
\qbezier(120,80)(120,80)(100,0)
\qbezier(120,80)(120,80)(105,0)
\qbezier(120,80)(120,80)(110,0)

\qbezier(120,80)(120,80)(115,0)
\qbezier(120,80)(120,80)(120,0)
\qbezier(120,80)(120,80)(125,0)

\qbezier(120,80)(120,80)(130,0)
\qbezier(120,80)(120,80)(135,0)
\qbezier(120,80)(120,80)(140,0)

%%% line: middle %%%
\qbezier(70,120)(70,120)(55,40)
\qbezier(70,120)(70,120)(70,40)
\qbezier(70,120)(70,120)(85,40)

%%% line: highest %%%
\qbezier(70,120)(70,120)(120,80)

\end{picture}
\end{center}
\caption{\small\it The output (after contractions): $E^+(v_0)$ is the optimal average edge-cut}
\label{FIG: output}
\end{figure}

\vskip 0.1 in
In the next section, we will prove that the algorithm provides an optimal solution. That is, we prove the following theorem as our main result.

\begin{theorem} \label{maintheorem}
The output $Q$ of the algorithm is an edge-cut of the input rooted tree $T$ with the average weight
$$\frac{\sum_{e \in Q}w(e)}{|Q|}$$ maximum among all edge-cuts separating the root $v_0$ and the set $L$
of leaves.
\end{theorem}

\section{Proof of Optimality}

\subsection{Lemmas}

Before the proof of Theorem~\ref{maintheorem}, we need the following lemmas.

\begin{lemma} \label{replacementlemma} Let $Q$ be a root-separating cut of $T$ and let $e \in Q $ but $e$ is not a leaf-edge of $T$. If $\lambda(e) > \frac{w(Q)}{|Q|} $, then $\exists$ $Q' \neq Q$ with $\frac{w(Q')}{|Q'|} > \frac{w(Q)}{|Q|}$. Specifically, $Q' = (Q \setminus \lbrace e \rbrace)\cup E^+(e)$.
\end{lemma}

\begin{proof} By the given conditions that $\lambda(e) > \frac{w(Q)}{|Q|} $,
 we have the following.
\begin{equation} \label{lambdaineq2}
  \frac{w(E^+(e)) - w(e)}{| E^+(e) | - 1} = \lambda(e) ~ > ~
   \frac{w(Q)}{|Q|} = \frac{w(Q \setminus \lbrace e \rbrace) + w(e)}{| Q \setminus \lbrace e \rbrace | + 1}
\end{equation}
We will use the following classical inequality
\begin{equation} \label{classic2}
\frac{a}{b} > \frac{c}{d}\implies \frac{a + c}{b + d} > \frac{c}{d}
\end{equation}
and define $a,b,c,d$ as follows.
\begin{eqnarray}\nonumber
a & = & w(E^+(e)) - w(e) \\\nonumber
b & = & | E^+(e) | - 1 \\\nonumber
c & = & w(Q \setminus \lbrace e \rbrace) + w(e) \\\nonumber
d & = & | Q \setminus \lbrace e \rbrace | + 1.
\end{eqnarray}
Then,
 with the above definitions,
 Inequality (\ref{lambdaineq2}) is the LHS of the implication (\ref{classic2}). The RHS of implication (\ref{classic2}) is as follows.
\begin{equation} \label{trade2}
\frac{w(Q \setminus \lbrace e \rbrace) + w(E^+(e))}{| Q \setminus \lbrace e \rbrace | + |E^+(e)|} >  \frac{w(Q \setminus \lbrace e \rbrace) + w(e)}{| Q \setminus \lbrace e \rbrace | + 1} = \frac{w(Q)}{|Q|}.
\end{equation}
Inequality (\ref{trade2}) says that the cut $Q' = (Q \setminus \lbrace e \rbrace)\cup E^+(e)$ has the average weight {\it greater than} the average weight of $Q$.
\end{proof}

\begin{defn} Let $Q$ be a root-separating cut of $T$ with $Q \neq E^+(v_0)$. Let $H$ be the component of $T-Q$ such that the root
 $v_0 \in V(H)$. For terminology we will refer to $H$ as {\it the subtree (of T) internal to Q}, or we may say $H$ {\it is the internal subtree of} \;$Q$.
\end{defn}

\begin{lemma} \label{internallem} Let $Q \neq E^+(v_0)$ be a root-separating cut of $T$ and let $H$ be the subtree internal to Q. Let $e \in H$ be a leaf edge of $H$. Then at least one of the following holds.
\begin{description}
\item[(i)] $\lambda(e) > \frac{w(Q)}{|Q|} $
\item[(ii)] $\exists$ $Q' \neq Q$ with internal subtree $H'$ such that $|E(H')| < |E(H)|$ and $\frac{w(Q')}{|Q'|} \geqslant \frac{w(Q)}{|Q|}$. Specifically, $Q' = (Q \setminus E^+(e))\cup \lbrace e \rbrace$.
\end{description}
\end{lemma}

\begin{proof} Suppose $\mathbf{(i)}$ is not true.
That is,
$\frac{w(Q)}{|Q|} \geq
  \lambda(e)$.
Thus, we have the following.
\begin{equation}
\label{lambdaineq}
\frac{w(Q \setminus E^+(e)) + w(E^+(e))}{|Q \setminus E^+(e)| + |E^+(e)|} = \frac{w(Q)}{|Q|} ~
\geqslant ~
\lambda(e) = \frac{w(E^+(e)) - w(e)}{|E^+(e)| - 1}
\end{equation}
We will use the following classical inequality.
\begin{equation} \label{classic1}
\frac{a + c}{b + d} \geqslant \frac{c}{d}\implies \frac{a}{b} \geqslant \frac{a + c}{b + d}
\end{equation}
Define $a,b,c,d$ as follows.
\begin{eqnarray}\nonumber
a & = & w(Q \setminus E^+(e)) + w(e) \\\nonumber
b & = & | Q \setminus E^+(e) | +1 \\\nonumber
c & = & w(E^+(e)) - w(e) \\\nonumber
d & = & | E^+(e) | - 1
\end{eqnarray}
Then with the above definitions, inequality (\ref{lambdaineq}) is the LHS of the implication (\ref{classic1}). The RHS of implication (\ref{classic1}) is as follows.
\begin{equation} \label{leaftrade}
\frac{w(Q \setminus E^+(e)) + w(e)}{| Q \setminus E^+(e) | +1} \geqslant  \frac{w(Q \setminus E^+(e)) + w(E^+(e))}{|Q \setminus E^+(e)| + |E^+(e)|} = \frac{w(Q)}{|Q|}
\end{equation}
Inequality (\ref{leaftrade}) says that the cut $Q' = (Q \setminus E^+(e))\cup \lbrace e \rbrace$ has the average weight {\it at least} the average weight of $Q$. Clearly if $H'$ is the internal subtree of $Q'$ then $|E(H')| < |E(H)|$. So condition $\mathbf{(ii)} $ holds as desired.
\end{proof}

\vspace{0.2 in}
\begin{lemma} \label{cqlemma} Let $e_1$ be the maximum edge in the contractibility ordering of $E(T - L)$. Assume that $\lambda(e_1) > \alpha_0$. Then an optimal solution of $T' = T/e_1$ is also an optimal solution of $T$.
\end{lemma}

\begin{proof} It is sufficient to show that {\it there exists} an edge cut $Q_0$ of $T$ which achieves the maximum average weight among all root-separating cuts of $T$ and $e_1 \notin Q_0$.\\

Let $Q_0$ be an optimal average weight cut in $T$. Assume $e_1 \in Q_0$. Observe that if $\frac{w(Q_0)}{|Q_0|} = \alpha_0 $ then $\lambda(e_1) > \frac{w(Q_0)}{|Q_0|} $ and by Lemma \ref{replacementlemma}
 we may define $Q' = (Q_0 \setminus \lbrace e_1 \rbrace) \cup E^+(e_1)$
 and we have $\frac{w(Q')}{|Q'|} > \frac{w(Q_0)}{|Q_0|}$, a contradiction to the optimality of $Q_0$. So $\frac{w(Q_0)}{|Q_0|} > \alpha_0 $ must be true,
and $Q_0 \neq E^+(v_0)$, and also
\begin{equation}
\label{EQ: 20180713-1}
\frac{w(Q_0)}{|Q_0|} \geqslant \lambda(e_1).
\end{equation}

Among all optimal average weight cuts $Q_0$ satisfying the above, we may assume $Q_0$ has the smallest possible internal subtree $H_0$. Note that $H_0 \neq \emptyset$. Then by Lemma \ref{internallem} there is a leaf edge $e \in H_0$ with
\begin{equation}
\label{EQ: 20180713-2}
\lambda(e) > \frac{w(Q_0)}{|Q_0|}.
\end{equation}
 But since $e \neq e_1$ we have $\lambda(e) > \lambda(e_1)$
(by Inequalities~(\ref{EQ: 20180713-1}) and (\ref{EQ: 20180713-2})),
 a contradiction to the definition of $e_1$ as the maximum edge in the contractibility ordering. So $e_1 \notin Q_0$ as desired.
\end{proof}

\subsection{Proof of Theorem \ref{maintheorem}}
 By Lemma~\ref{cqlemma}, it is sufficient to show that, after all possible contractions of edges, $E^+(v_0)$ is an optimal solution in the resulting tree. \\

Suppose on the contrary that the output of the algorithm is not the proposed cut.  Let $Q$ be an optimal cut with an internal subtree $H$ with $|H|$ as small as possible. Suppose that $Q \neq E^+(v_0)$. That is,
\begin{equation}
\label{EQ: 20180713-B1}
\frac{w(Q)}{|Q|} > \frac{w(E^+(v_0))}{|E^+(v_0)|} = \alpha_0.
\end{equation}

Apply Lemma \ref{internallem} here.
The case (ii) of Lemma \ref{internallem} does not occur since the cut $Q$ is optimal. Hence the case (i) implies the existence of a leaf $e$  in the internal subtree $H$ with
\begin{equation}
\label{EQ: 20180713-B2}
\lambda(e) > \frac{w(Q)}{|Q|}.
\end{equation}
By Inequalities~(\ref{EQ: 20180713-B1}) and (\ref{EQ: 20180713-B2}), the leaf $e$ is a contractible edge. Moreover, the maximum edge in the contractibility ordering is contractible, contradicting the assumption that the algorithm terminated.

\section{Conclusions and Remarks}

\subsection{Computational complexity of Algorithm}
In addition, we find that the computational complexity of the algorithm is polynomial and so its implementation will be efficient and scalable to large data sets. Let $|V(T-L)|=n$. The cost of step 1 is a constant. Step 2 and Step 3 are repeated at most $n$ times.
In the first loop of repeating,
the cost of step 2 is $O(n(\log_2 n))$ or $O(n^2)$
 since the sorting of $\lambda$'s is
 involved, while the cost of step 2 in each loop
 after the first sorting is at most $O(n)$
 since we only need to place one item
 $\lambda(e^*)$ to a proper position
in a well sorted sequence. The cost of step 3 in every loop is at most $O(n)$ since only updating of the graph and the revising the weight $\lambda$ are involved.
Hence, the cost of the worst case has an upper bound of $O(n^2)+n(O(n)+O(n))=O(n^2)$.

\subsection{Conjectures and open problems}

\begin{problem}
\label{PROB: average max}
Let $G$ be a  network with the source $s$ and the sink $t$ and
  associated with a weight $w: E(G \rightarrow R^+$.
Find an edge-cut $T$ separating $t$ from $s$ such that the average weight of $T$ is maximum among all such edge-cuts.
\end{problem}

The problems of finding cuts with total weights maximized or minimized
 have been well studied.
One is known as an NP-complete problem, while another is polynomial
\cite{Bondy1976,Cormen2009,Goldberg1988}.
But few study have been done yet for finding an edge-cut with the average weight maximized (or minimized).
We would like to propose the following conjecture.

\begin{conjecture}
 Problem~\ref{PROB: average max} is an NP-complete problem.
\end{conjecture}

\bibliographystyle{plain}
\bibliography{avgcutbib}

\end{document}